\documentclass[reqno,a4paper]{article}

\pdfoutput=1
\usepackage[top=1in, bottom=1.25in, left=1in, right=1in]{geometry}

\usepackage[ngerman,british]{babel}

\usepackage[utf8]{inputenc}
\usepackage[normalem]{ulem}

\usepackage{soul}

\usepackage[T1]{fontenc}
\usepackage{lmodern}
\usepackage{amsmath,amssymb,amsthm,amsfonts}
\usepackage{mathtools}
\mathtoolsset{showonlyrefs}
\usepackage{nicefrac}       
\usepackage{braket}
\usepackage{xspace}
\usepackage{float}
\restylefloat{table}
\usepackage{booktabs}
\usepackage{bibentry}
\usepackage{authblk}
\usepackage{subcaption}
\usepackage{graphicx}
\usepackage{color}
\usepackage{epstopdf}
\DeclareGraphicsExtensions{.pdf,.eps,.png}
\usepackage{stackengine}

\usepackage{paralist}
\usepackage{enumitem}

\usepackage{bbm}

\usepackage{calc}

\usepackage{accents}

\newcommand{\Z}{\mathbb{Z}}

\newcommand{\R}{\mathbb{R}}

\soulregister{\ref}{1}
\soulregister{\cite}{1}
\soulregister{\pageref}{1}
\soulregister{\eqref}{1}

\newcommand{\blue}[1]{#1}

\newtheorem{theorem}{Theorem}[section]

\newtheorem{definition}[theorem]{Definition}

\newtheorem{lemma}[theorem]{Lemma}

\newcommand{\qm}[1]{``#1''}

\newcommand{\eg}{e.g.\ }
\renewcommand{\set}[1]{\{#1\}}

\makeatletter
\newenvironment{varsubequations}[1]
 {%
  \addtocounter{equation}{-1}%
  \begin{subequations}
  \renewcommand{\theparentequation}{#1}%
  \def\@currentlabel{#1}%
 }
 {%
  \end{subequations}\ignorespacesafterend
 }
\makeatother

\DeclarePairedDelimiter{\ceil}{\lceil}{\rceil}

\usepackage[bookmarksnumbered,bookmarksopen,unicode,colorlinks,allcolors=purple]{hyperref}
\usepackage{url}   

\usepackage[]{todonotes}

\title{Quantum Subroutines in Branch-Price-and-Cut for Vehicle Routing}

\author[1,2,*]{Friedrich Wagner}
\author[1]{Frauke Liers}
\affil[1]{Department of Data Science, University of Erlangen-Nürnberg}
\affil[2]{Fraunhofer Institute for Integrated Circuits, Nürnberg}
\affil[*]{\texttt{\small friedrich.wagner@iis.fraunhofer.de}}

\begin{document}
	
	\maketitle
	
	\begin{abstract}
	Motivated by recent progress in quantum hardware and algorithms,
	researchers have developed quantum heuristics for optimization problems, aiming for advantages over classical methods.
	To date, quantum hardware is still error-prone and
	limited in size such that quantum heuristics cannot be scaled to relevant problem sizes
	and are often outperformed by their classical counterparts.
	Moreover, if provably optimal solutions are desired, one has
	to resort to classical exact methods.
	As, however, quantum technologies may improve considerably in future, 
	we demonstrate in this work how quantum heuristics with
	limited resources can be integrated in large-scale exact
	optimization algorithms for NP-hard problems.
	To this end, we consider vehicle routing as a
	prototypical NP-hard problem.
	We model the pricing and separation subproblems arising in a branch-price-and-cut algorithm as quadratic unconstrained binary optimization problems.
	This allows to use established quantum heuristics like quantum annealing or the quantum approximate optimization algorithm for their solution.
	A key feature of our algorithm is that it profits not only from the best solution returned by the quantum heuristic but from all solutions below a certain cost threshold,
	thereby exploiting the inherent randomness is quantum algorithms.
	Moreover, we reduce the requirements on quantum hardware since the subproblems, which are solved via quantum heuristics, are smaller than the original problem.
	We provide a \blue{proof-of-concept} experimental study comparing quantum annealing to simulated annealing and to established classical algorithms in our framework.
	While our hybrid quantum-classical approach is still
	outperformed by purely classical methods, our results reveal
	that both pricing and separation may be well suited for quantum heuristics
	if quantum hardware improves.
\end{abstract}

\section{Introduction}
Combinatorial optimization problems arise naturally in various fields of industry and society.
Although being NP-hard in general,
modern algorithms and implementations routinely solve even large instances in reasonable time~\cite{Bixby2002,Koch2022}.
Nonetheless, there exist relevant problem instances which remain intractable for state-of-the-art methods~\cite{Gleixner2021}.
As a result, recent progress in quantum hardware and the discovery of quantum algorithms
with provable advantages have motivated researchers to develop quantum approaches for such problems~\cite{Abbas2023,bochkarev2024}.
However, exact quantum algorithms for combinatorial optimization by far exceed the capabilities of existing quantum hardware~\cite{ammann2023realisticruntimeanalysisquantum}.
Therefore, researchers focus on quantum heuristics, which are less hardware-demanding than exact approaches.
Popular examples are quantum annealing (QA)~\cite{Albash2018,McGeoch2023} and the quantum approximate optimization algorithm (QAOA)~\cite{Farhi_2014,Harrigan_2021}.
It remains open whether future quantum algorithms and hardware will outperform their classical counterparts.
Currently, the size limitation and noise vulnerability of existing quantum devices hinder the scaling of quantum heuristics to relevant problem sizes~\cite{Guerreschi2019,JuengerLobe_2021}.
Moreover, when aiming at provable optimality, classical exact methods are inevitable.
This work aims at bringing quantum optimization closer to practical utility by integrating quantum subroutines in established integer optimization methods.
	
In particular, we employ quantum annealing as a pricing and separation heuristic in a branch-price-and-cut algorithm for the capacitated vehicle routing problem (CVRP).
This hybrid approach bears two major advantages compared to solving the CVRP directly via QA.
First, larger instances can be handled since the subproblems solved via QA are considerably smaller than the original problem.
Nonetheless, both subproblems are NP-hard such that we cannot expect to solve them efficiently via classical algorithms.
Second, multiple solutions \blue{of different quality} to a single subproblem are valuable for the overall algorithm.
QA typically delivers several thousand solutions in less than one second~\cite{Tassef2022}.
In a direct application of QA, all but the best solution are discarded.
In contrast, our \blue{hybrid} branch-price-and-cut approach benefits from \blue{all} solutions \blue{below a certain cost threshold}.
Although CVRP is NP-hard, even large instances can be solved to global
optimality by modern
branch-(price-)and-cut approaches~\cite{Lysgaard_2004,Pecin2017,Pessoa2019}.
We thus expect that our hybrid
algorithm is outperformed by classical methods when using current quantum hardware.
This is also confirmed
by our experiments. Nevertheless, quantum
annealing may improve in the future such that the hybrid algorithm is not only interesting from an algorithmic point of view,
but in the future also has the potential to be beneficial in practice. 
To summarize, this work demonstrates how limited quantum resources can help to solve problems of relevant size if quantum hardware advances.

The CVRP is an archetypical combinatorial optimization problem which is highly relevant in applications~\cite{Toth2002,Toth2014}.
Given a set of customers with individual demands and a depot with a fleet of trucks of limited capacity,
the CVRP asks for an assignment of routes to trucks such that each customer is visited and no route exceeds the capacity.
The objective is to minimize the aggregate route length.
State-of-the-art exact algorithms for the CVRP are mostly branch-price-and-cut approaches, which can solve instances with several hundreds
of customers to optimality in the order of hours~\cite{Pecin2017,Pessoa2019}.
Such a branch-price-and-cut algorithm alternatingly solves a master problem and a subproblem, called the \emph{pricing problem}.
This procedure is also known as \emph{column generation}.
While the master problem is a linear program, which can be solved efficiently, the pricing problem is
typically NP-hard. This is also the case in this work.
Apart from column generation, branch-price-and-cut utilizes cutting plane separation to speed up the algorithm.
Unlike pricing, separation is not ultimately required for the algorithm to terminate but often speeds up the solution process by several orders of magnitude~\cite{MARCHAND2002}.
Applying QA requires modeling both the pricing and the separation
subproblem as quadratic unconstrained binary optimization (QUBO).
In this work, we develop {such} QUBO models{, aiming at low dimension and density}, which is crucial for a successful application of existing quantum hardware.
We then integrate the pricing and separation models in \blue{subroutines of} a
branch-price-and-cut algorithm for CVRPs.
Finally, we evaluate the integrated \blue{subroutines} on actual quantum hardware.
To this end, we benchmark pricing and separation via QA against simulated-annealing and established integer programming methods on several instances from literature and form a real-world application.
Our results indicate that QA can be a promising heuristic for pricing and separation if quantum hardware advances.

\paragraph{Related work.}
The application of quantum annealing to vehicle routing problems (VRPs) has been studied in literature before.
In~\cite{Irie2019}, a time-scheduled variant of the VRP is directly solved on quantum annealing hardware.
Due to the size limitations of current quantum hardware, the authors consider only toy examples.
On the contrary, the authors of~\cite{Crispin2013,Syrichas2017} employ a purely classical simulation of QA to heuristically solve VRPs of practically relevant size.
Their results show advantages of the classically simulated quantum-annealing algorithm over a standard simulated-annealing algorithm.
Simulated annealing is an established heuristics for combinatorial optimization~\cite{Bertsimas1993}, which is often used as a benchmark for quantum annealing
since it is conceptually similar.
Ref.~\cite{Borowski2020} proposes a hybrid quantum-classical heuristic for VRPs, which is based on QA.
Their experiments indicate that the proposed hybrid method may outperform classical heuristics.
Ref.~\cite{Harikrishnakumar2020} develops a QUBO model for a multi-depot variant of the CVRP.
However, no evaluation of the proposed QUBO model on actual quantum hardware is performed.
The authors of~\cite{Feld2019} apply a well-known heuristic decomposition of the CVRP
into a knapsack problem and several traveling salesperson problems.
They compare solving either all or only some of these subproblems via QA.
Their experiments do not show a clear advantage of QA over classical heuristics.
Ref.~\cite{Bao2021} solves a variant of the VRP by a QA-inspired, classical heuristic.
An experimental evaluation shows advantages of the proposed method over two non quantum-inspired heuristics.
In~\cite{Tambunan2022}, a multi source-destination VRP is solved by QA.
Computational tests show an advantage of QA over simulated annealing.
Crucially, all studies mentioned so far give \emph{heuristic} solutions to VRPs.
On the contrary, we integrate QA in an \emph{exact} algorithm
which solves the CVRP to proven optimality.

Similar to the application of QA to VRPs, the integration of quantum heuristics in classical exact decomposition algorithms is a field of active research.
Ref.~\cite{Franco2022} proposes to use QA for solving the master problem in a Benders decomposition
for specific mixed-integer programs arising in the context of neural network certification.
The authors provide an experimental evaluation which shows that the Benders subproblem often exceeds the resources of currently available quantum hardware.
The follow-up work~\cite{Franco2023} proposes to use QA also for the pricing subproblem in a column generation approach.
Their results show that the column generation subproblem is more suited for quantum approaches than the Benders subproblem.
The authors of~\cite{OssorioCastillo2022} apply QA to the pricing problem in a column generation approach for a process scheduling problem.
However, they do not compare their algorithm against classical methods.
Similarly, Ref.~\cite{Hirama2023} proposes a column generation approach to constrained quadratic binary programs which solves the pricing step via QA.
They compare quantum annealing to simulated annealing and a branch-and-cut solver applied directly to the constrained quadratic binary program.
Their results do not show a clear advantage of QA over classical algorithms.
Related, the authors of~\cite{da_Silva_Coelho_2023} apply QA to the pricing problem in a column generation algorithm for graph coloring.
Their numerical experiments indicate advantages of QA over classical pricing methods.
The authors of~\cite{Svensson_2023} employ a variant of QAOA as a primal heuristic in a branch-and-price algorithm.
Similarly, Ref.~\cite{wagner2023enhancing} develops a quantum-classical heuristic for QUBO which is well-suited for integration in branch-and-cut algorithms.
In contrast to the studies mentioned so far,
we employ QA not only for pricing but also as a separation heuristic.
Moreover, we consider the CVRP as an example problem,
for which branch-price-and-cut approaches have proven particularly successful.
We compare QA to simulated annealing and to established pricing and separation algorithms from literature.

Owing to its practical relevance and inherent hardness,
the research community has developed a great number of classical algorithms for the vehicle routing problem.
As a result, modern algorithms can solve instances with several hundred customers to proven optimality in the order of hours.
For a comprehensive survey, we refer to~\cite{Toth2014}.
Noteworthy, several state-of-the-art exact algorithms rely on branch-price-and-cut~\cite{Pecin2017,Pessoa2019}.
The branch-price-and-cut approach considered in this work was originally introduced in~\cite{Bixby1999}.

\paragraph{Our contribution.}
We are not aware of any work that integrates QA in a branch-price-and-cut algorithm for the CVRP and,
in particular, to use QA as a separation heuristic.
We develop novel QUBO models for both the pricing and separation subproblem.
Our models are {often} small and sparse, which renders them well-suited for QA.
As a result, our approach solves CVRP instances which are larger than the capabilities of existing quantum hardware.
We provide an experimental study on non-trivial instances from literature and from a real-world application.
Therein, we compare QA to simulated annealing and to established algorithms for pricing and separation.
Our results show that both pricing and separation can be suitable problems for quantum heuristics in future.

\paragraph{Structure.}
The remainder of this work is organized as follows.
Section~\ref{sec:qvrp_prelim} introduces the prerequisites necessary for the upcoming sections.
In Section~\ref{sec:qvrp_algo}, we introduce the branch-price-and-cut algorithm for the capacitated vehicle routing problem.
Section~\ref{sec:qvrp_qubos} develops the QUBO models for the pricing and separation subproblems.
In Section~\ref{sec:qvrp_exp}, we evaluate the proposed methods on actual quantum hardware.
Finally, we end with a conclusion and discuss open questions for future research in Section~\ref{sec:qvrp_concl}.

\section{Preliminaries}\label{sec:qvrp_prelim}
Before stating our main results, we introduce quadratic unconstrained binary optimization, quantum annealing
and branch-price-and-cut.
Furthermore, we formally define the capacitated vehicle routing problem as a combinatorial optimization problem.

\paragraph{Quadratic unconstrained binary optimization.}
Most implementations of quantum-annealing or QAOA require the problem of interest to be modeled as a quadratic unconstrained binary optimization (QUBO),
which are problems of the form
\begin{align}\label{eq:qubo}
	\min_{x\in \{0,1\}^N} x^TQx,
\end{align}
where $Q \in \R^{N,N}$ is a real-valued matrix.
Without loss of generality, we assume $Q$ to be upper-triangular throughout this work.
We note that $x^2=x$ if $x\in \{0,1\}$.
Thus, the diagonal of $Q$ implicitly encodes linear terms.
We call terms of the form $x_i x_j$ with $i\neq j$ \textit{quadratic} terms although, strictly speaking, they are bilinear.
For a given upper-triangular matrix $Q$,
we define its \textit{density} as
the fraction of non-zero off-diagonal elements
\begin{align}\label{eq:def_density}
	\rho(Q)\coloneqq \frac{\sum_{i<j}\mathbbm{1}({q_{ij}\neq0})}{N(N-1)/2}\in [0,1]\,.
\end{align}
Being equivalent to the maximum cut problem, QUBO is NP-hard~\cite{Hammer_1965}.
Thus, polynomial-time algorithms exist for special cases only~\cite{Groetschel_1984,Hadlock_1975,Liers_2012}.
QUBO has been extensively studied in literature, for
comprehensive books we refer to~\cite{Punnen_2022,Deza1997}.

Most optimization problems include constraints, which is also the case
in this work.
QUBO with additional constraints can still be solved directly by classical branch-and-cut algorithms~\cite{Buchheim_2010}.
{Although quantum algorithms exist which can handle constraints~\cite{Hadfield_2019,gilliam2021grover},
practical implementations are mostly limited to pure QUBOs.}
To transform a constrained problem into an equivalent QUBO problem, all constraints are incorporated into the cost function $C(x)\coloneqq x^TQx$ via penalty terms.
In general, such a penalized QUBO cost function is of the form
\begin{align}\label{eq:qubo_costfunction}
	C(x) = C_1(x) + P \cdot C_2(x)\,.
\end{align}
Here, $C_1(x)$ is a quadratic function that, if $x\in \{0,1\}^N$ satisfies all constraints, models the costs of the given constrained problem.
$P>0$ is a penalty factor and $C_2(x)$ is a \blue{non-negative} quadratic function which penalizes infeasibility.
It is constructed such that
$C_2(x)=0$ if and only if $x$ encodes a feasible solution, that is, $x$ satisfies all constraints.
Hence, for a sufficiently large value of $P$, an optimal solution to the QUBO defined by~\eqref{eq:qubo_costfunction} is feasible and thus optimal to the constrained problem.

\paragraph{Quantum annealing.}
Quantum annealing is a heuristic algorithm for optimization problems.
It is typically implemented on analog quantum computers built exclusively for this purpose.
This is in contrast to universal, gate-based quantum computers,
which are built to execute arbitrary quantum algorithms.
{The quantum annealer used in this work has} roughly 50 times more qubits than gate-based quantum devices~\cite{Dwave_advantage,IBMQuantum}.
Still, even with over 5000 physical qubits, their low connectivity {usually} requires an embedding of each logical QUBO variable into several physical qubits,
a process called \emph{minor embedding}~\cite{Choi2008,Choi2010}.
This strongly limits the size of solvable problems if the QUBO matrix is dense.
For example, the quantum device used in our experiments has $5760$ qubits.
However, for a dense QUBO matrix, the problem size is limited to $N\leq 169$ variables~\cite{Dwave_advantage}.
It is crucial keep the number of required qubits as small as possible
since the performance of QA hardware decreases strongly with the number of used qubits~\cite{McGeoch2023}.
From a modeling perspective, this requirement translates to keeping both the model dimension $N$ and the model density $\rho$ as low as possible.

The anticipated advantage of quantum annealing is that it quickly generates many high-quality solutions.
Typically, several thousand solutions are retrieved in less than one second~\cite{Tassef2022}.
{If the practical increase in runtime will not be too high for larger problems},
quantum annealing might outperform classical heuristics for large problems at some point in future.
We however note that to date, classical heuristics still outperform quantum annealing, which is also confirmed by our experiments~\cite{bochkarev2024,JuengerLobe_2021}.

When applying quantum annealing directly to an application problem, usually all generated solutions except for the best one are discarded.
The main idea of this work is to apply quantum annealing to problems for which more solutions than only the best one are valuable.
To this end, we consider two subproblems in a classical decomposition algorithm for vehicle routing and solve them via quantum annealing.
The overall algorithm benefits from all subproblem solutions below a certain cost threshold.
Moreover, the decomposition algorithm solves the input problem to proven optimality,
which is not {the} case when using QA as a standalone algorithm.

\paragraph{Branch-price-and-cut.}
We consider the established decomposition algorithm \emph{branch-price-and-cut}, which we briefly explain here.
For a comprehensive introduction, we refer to~\cite{Desrosiers2011}.
Branch-price-and-cut is a technique for solving (mixed-)integer linear programs (IPs),
which are problems of the form
\begin{subequations}\label{eq:mip_generic}
	\begin{alignat}{6}
		& \min_{x}           & c^Tx        &             &  & \label {eq:mip_generic_objective}                         \\
		& \mathrm{s.t.}\quad & Ax & \geq b      &  & \label{eq:mip_generic_con}           \\
		&                    & x                                    & \in \Z^p\times\R^q &  & \label{eq:mip_generic_int}
	\end{alignat}
\end{subequations}
where $c \in \R^n$, $n \in \Z_{>0}$, $A \in \R^{m\times n}$, $m \in \Z_{>0}$, $b \in \R^m$ and $p,q \in \Z_{\geq0}$ such that $p+q = n$. \noeqref{eq:mip_generic_objective,eq:mip_generic_con,eq:mip_generic_int}
Branch-price-and-cut is tailored to IPs with a (exponentially) large number of variables $n$ and a (polynomially) small number of constrains $m$.
It is based on column generation,
which in turn is a technique for solving the linear programming (LP) relaxation of~\eqref{eq:mip_generic},
\begin{subequations}\label{eq:lp_generic}
	\begin{alignat}{6}
		& \min_{x}           & c^Tx        &             &  & \label {eq:lp_generic_objective}                         \\
		& \mathrm{s.t.}\quad & Ax & \geq b      &  & \label{eq:lp_generic_con}           \\
		&                    & x                                    & \in \R^n &  & .\label{eq:lp_generic_int}
	\end{alignat}
\end{subequations}
The optimum value of~\eqref{eq:lp_generic} is a lower bound on the optimum value of~\eqref{eq:mip_generic} since the solution set of~\eqref{eq:lp_generic} is a superset of the solution set of~\eqref{eq:mip_generic}.
Column generation only considers a subset of variables and solves the LP~\eqref{eq:lp_generic} over this restricted variable set~\cite{Desrosiers2005,Luebbecke2005}.\noeqref{eq:lp_generic_objective,eq:lp_generic_con,eq:lp_generic_int}
Here, the rationale is that there exists an optimum solution in which the number of non-zero variables is bounded by the column rank of the constraint matrix, which, in turn, is bounded by the number of constraints $m$.
Thus, if the number of constraints is much smaller than the number of variables, there exist an optimum solution in which almost all variables attain the value zero.
After solving the restricted LP, a subproblem is invoked which checks if additional variables exist which could improve the solution and should thus be included in the restricted LP.
If no such variables exist, the optimum solution of the restricted LP is indeed optimal for the original LP.
This subproblem is known as \emph{pricing}.

Column generation can be extended to IPs.
To this end, it is integrated in a branch-and-bound algorithm, a technique known as branch-and-price~\cite{Barnhart98,Luebbecke2005}.
Most simplified, branch-and-bound consists of two routines.
The first routine computes lower and upper bounds on the value of an optimal solution in a given subset of solutions to the IP.
In this work, this is done by solving linear relaxations via column generation.
The second routine partitions the IP solution set into disjoint subsets, called branches.
For example, we can define two branches by fixing a binary variable to $0$ and $1$, respectively.
In the beginning, branch-and-bound calculates a lower and upper bound for the entire solution set.
If there is a gap between upper and lower bound, the algorithm branches.
For each branch, we again compute a lower and upper bound, keeping track of a globally valid lower and upper bound.
As long as there is a gap between the global upper and lower bound, the algorithm continues to branch, building up a tree as illustrated in Fig.~\ref{fig:concept_bpc} (left).
Importantly, if the lower bound in a branch exceeds the global upper bound, the branch can be excluded from the search.
Thereby, we try to keep the width of the search tree tractable, which can be exponentially large in the worst case.
\begin{figure}
	\centering
    \includegraphics[height=6cm]{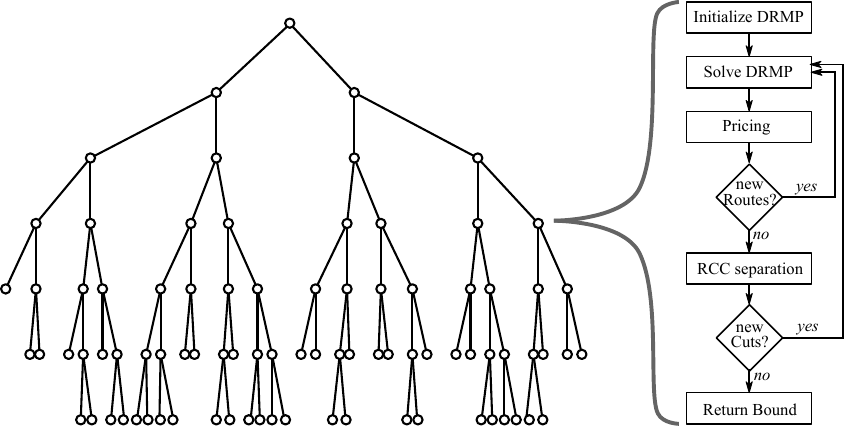}
    \caption{Schematic workflow of branch-price-and-cut. The algorithm builds a branch-and-bound tree (left).
	Each node represents a subset of solutions, where the root node corresponds to the entire solution set.
	The two children of a node divide its solution subset further into two disjoint subsets.
	For each node, we calculate a lower bound via column generation and cutting plane separation (right).
	As long as global upper and lower bound differ, the algorithm continues to branch (edges in the tree).
	If the local lower bound in a node exceeds the global upper bound, the node can be excluded from the search (nodes without children).
	We propose to use quantum annealing as a heuristic for pricing and separation.	\label{fig:concept_bpc}}
\end{figure}

Thus, the runtime of a branch-and-bound algorithm heavily depends on the tightness of the employed bounding method.
An established method to strengthen the linear programming bound is to add additional inequalities to the LP
which are valid for the integer case but cut off fractional points in the linear relaxation.
Such inequalities are known as \emph{cutting planes}.
The problem of constructing a cutting plane which cuts off a given fractional solution is known as the \emph{separation problem}.
In general, the separation problem can be NP-hard, which is also the
case for the class of cutting planes considered here.
Thus it is typically solved via heuristics.
Integrating cutting plane separation in a branch-and-price algorithm is called
branch-price-and-cut~\cite{Barnhart2000,Juenger2000,Elf2001}.
Details of our branch-price-and-cut algorithm will be explained in Section~\ref{sec:qvrp_algo}.

Here, we propose to use quantum annealing as both a pricing heuristic and a separation heuristic.
As an example problem, we consider the CVRP.
The motivation is that state-of-the-art exact algorithms for the CVRP
are often branch-price-and-cut approaches~\cite{Toth2014,Pecin2017,Pessoa2019}.

\paragraph{The vehicle routing problem.}
The vehicle routing problem is an archetypical combinatorial optimization problem,
originally introduced by Dantzig and Ramser in 1959~\cite{Dantzig1959}.
Generalizing the prominent NP-hard travelling salesperson problem, it developed to one of the most studied combinatorial optimization problems due to both its practical relevance and inherent difficulty.
Many variants are studied in literature, for a comprehensive survey we refer to~\cite{Toth2002}.
Here, we focus on a basic variant, the {capacitated vehicle routing problem}.
However, the methods we develop in this work can readily be applied to other VRP variants as well.

Informally speaking, in the CVRP we are given a set of customers as well as a number of trucks located at a depot.
Each customer has an individual demand of goods, while the trucks are identical and have a limited capacity.
Now, the task is to assign a route to each truck such that each customer is visited at least by one truck and each route does not exceed the truck capacity.
More formally, the CVRP is defined as the following combinatorial optimization problem.
\begin{definition}[CVRP]\label{def:cvrp}
	An instance of the capacitated vehicle routing problem is given by a quadruple $(V^0,t,D,K)$ where
	\begin{itemize}
		\item $V^0 = \set{0,1\dots n}$ is an index set with $0$ denoting a depot and $\set{1\dots n}\eqqcolon V$ denoting a set of customers,
		\item $t:V^0 \times V^0\rightarrow \R_0^+$ is a metric satisfying the triangle-inequality,
		\item $D=\set{d_{{v}} \mid {{v}} \in V}$ are positive demands of customers and
		\item $K>0$ is the vehicle capacity.
	\end{itemize}
	We define a route $r = (v_1,v_2,\dots,v_{|r|})$, $v_i \in V$, as a sequence of customers.
	A feasible solution is given by a set of routes $R=\{r_1,r_2,\dots,r_{|R|}\}$
	such that each route does not exceed the capacity $K$
	and each customer is contained in at least one route.
	The cost of a feasible solution $R$ is given by its total length
	\[
	c = \sum_{r\in R} c_r
	\]
	where
	\[
	c_r = t(0,v_1)+ \sum_{i=1}^{|r|-1}t(v_i,v_{i+1}) + t(v_{|r|},0)
	\]
	is the length of route $r\in R$.
	The goal is to minimize the total length.
\end{definition}
{In this work, we restrict the distances, demands and the capacity to integer values, which is common in the CVRP literature~\cite{Uchoa2017} \blue{and required for the QUBO transformation}.}
Having introduced the necessary prerequisites, we now detail the proposed branch-price-and-cut algorithm.

\section{Branch-Price-and-Cut Algorithm}\label{sec:qvrp_algo}
Our work builds upon the branch-price-and-cut algorithm originally introduced in~\cite{Bixby1999}.
We remark that more advanced branch-price-and-cut algorithms exist~\cite{Pecin2017,Pessoa2019}.
However, we do not aim at improving the state-of-the-art for exact vehicle routing solvers but at showcasing how quantum heuristics can be beneficial for exact integer optimization algorithms.
Here, the algorithm of~\cite{Bixby1999} is particularly well-suited to illustrate the integration of quantum subroutines in branch-price-and-cut.
{We emphasize that we do not implement the entire branch-price-and-cut algorithm but focus on specific parts that are relevant for the proposed quantum subroutines.}
{The algorithm works on} an exponentially large formulation of the CVRP,
known as the \emph{set cover model}, see e.g.~\cite{Toth2002}.

\paragraph{Set cover model.}
Let $(V^0,t,D,K)$ be a CVRP instance.
Denote by $\mathcal{R}$ the set of all routes with a demand not exceeding $K$.
For a customer ${v}\in V$ and a route $r \in \mathcal{R}$,
we set $\alpha_{{v}r} = 1$ if $i\in r$, otherwise we set $\alpha_{{v}r} = 0$.
With this notion, the CVRP can be modeled by
\begin{varsubequations}{\blue{CVRP}}\label{eq:mip}
	\begin{alignat}{6}
		 & \min_{y}           & \sum_{r \in \mathcal{R}}c_r y_r        &             &  & \label {eq:mip_Objective}                         \\
		 & \mathrm{s.t.}\quad & \sum_{r \in \mathcal{R}}\alpha_{{v}r}y_r & \geq 1      &  & \quad \forall {v} \in V\label{eq:demands}           \\
		 &                    & y_r                                    & \in \{0,1\} &  & \quad \forall r\in\mathcal{R}\ .\label{eq:mip_biny}
	\end{alignat}
\end{varsubequations}
The objective~\eqref{eq:mip_Objective} minimizes the aggregate route length 
while constraints~\eqref{eq:demands} enforce that each customer is visited.\noeqref{eq:mip_Objective,eq:mip_biny}
We remark that it is never favorable to visit a customer more than once since the distances satisfy the triangle inequality.

At its core, branch-price-and-cut solves the continuous relaxation of~\eqref{eq:mip}, called master problem (MP),
\begin{varsubequations}{MP}\label{eq:mp_cvrp}
	\begin{alignat}{6}
		 & \min_{y}           & \sum_{r \in \mathcal{R}}c_r y_r        &           &  & \label {eq:mp_Objective}                         \\
		 & \mathrm{s.t.}\quad & \sum_{r \in \mathcal{R}}\alpha_{{v}r}y_r & \geq 1    &  & \quad \forall {v} \in V\label{eq:mp_demands}       \\
		 &                    & y_r                                    & \in [0,1] &  & \quad \forall r\in\mathcal{R}\ .\label{eq:mp_biny}
	\end{alignat}
\end{varsubequations}
It is well-known (see \eg\cite{Toth2014}) and also confirmed by our experiments, presented in Section~\ref{sec:qvrp_exp},
that the relaxation~\eqref{eq:mp_cvrp} typically yields tight bounds.\noeqref{eq:mp_Objective,eq:mp_demands,eq:mp_biny}
However, it contains an exponentially large number of variables.
Therefore, we solve~\eqref{eq:mp_cvrp} by column generation.
To this end, we consider the \emph{dual} of the master problem~\eqref{eq:mp_cvrp}.
The dual of a linear program is an equivalent linear program which has a variable (constraint) for each constraint (variable) in the original problem.
For an introduction to duality theory, we refer to~\cite{Padberg1999}.
The dual master problem (DMP) reads
\begin{varsubequations}{DMP}\label{eq:cvrp_dmp}
	\begin{alignat}{6}
		 & \max_{\pi}         & \sum_{{v} \in V}\pi_{{v}}           &            &  & \label {eq:obj_dmp}                             \\
		 & \mathrm{s.t.}\quad & \sum_{{v} \in V}\alpha_{{v}r}\pi_i & \leq c_r   &  & \quad \forall r \in \mathcal{R}\label{eq:r_dmp} \\
		 &                    & \pi_{{v}}                       & \in \R_0^+ &  & \quad \forall {v}\in V\ .\label{eq:binpi}
	\end{alignat}
\end{varsubequations}
Therein, $\pi_{{v}}$ is the dual variable associated to constraint~\eqref{eq:mp_demands} for customer ${v}\in V$.\noeqref{eq:obj_dmp,eq:binpi,eq:r_dmp}

\paragraph{The pricing problem.}
We consider an arbitrary subset of routes $\mathcal{R}'\subset \mathcal{R}$ and replace $\mathcal{R}$ by $\mathcal{R}'$ in \blue{~\eqref{eq:mp_biny} and thus also in~\eqref{eq:r_dmp}}.
We denote the resulting problem\blue{s} as \blue{restricted master problem (RMP) and} {dual restricted} master problem ({DRMP}), \blue{respectively}.
Let $\blue{\pi^*}$ be an optimal solution for the {DRMP}.
Now, the pricing problem asks whether $\blue{\pi^*}$ is feasible and thus optimal for the unrestricted \eqref{eq:cvrp_dmp}.
That is, we solve
\begin{align}\label{eq:pricing}
	\bar{c}_{\mathrm{min}} \coloneqq \min_{} \big\{ \underbrace{c_r -\sum_{v\in r}\blue{\pi^*_v} }_{\eqqcolon \bar{c}_r} \mid r\in \mathcal{R}\big\} \tag{P}
\end{align}
and check if $\bar{c}_{\mathrm{min}} \geq 0$.
If so, $\blue{\pi^*}$ is feasible and thus optimal for the unrestricted DMP~\eqref{eq:cvrp_dmp}.
Otherwise, we add a route $r$ with $\bar{c}_{r} < 0$ to the \blue{subset $\mathcal{R'}\subset\mathcal{R}$ in} the {DRMP} and repeat the procedure until $\bar{c}_{\mathrm{min}} \geq 0$.
We refer to the quantity $\bar{c}_{r}$ as the \emph{reduced cost} of a route $r \in \mathcal{R}$.

The pricing problem~\eqref{eq:pricing} is an NP-hard combinatorial optimization problem.
It asks for a route with minimal costs such that the route does not exceed the capacity.
Herein, the costs of a route are given by its length minus the prices of visited customers.
This problem is sometimes called the \emph{capacitated price collecting traveling salesperson problem} (CPCTSP)~\cite{Bixby1999}.
{In contrast to well-known traveling salesperson problem (TSP), CPCTSP has a capacity constraint and customers may be excluded from the route.}
As the name suggests, the CPCTSP can be shown to be NP-hard via a reduction of the TSP.
{To reduce TSP to CPCTSP, one chooses sufficiently large prices such that it is never favorable to exclude a customer.}

Importantly, an optimal solution to the CPCTSP~\eqref{eq:pricing}, which corresponds to a route with minimal reduced cost, does not necessarily lead to the largest cost reduction when added to the {DRMP}.
This is because the reduced cost $\bar{c}_{r}$ of a route $r\in \mathcal{R}$ quantify the cost reduction per unit increase of the associated variable $y_r$ in the MP~\eqref{eq:mp_cvrp}.
However, the maximum possible increase of $y_r$ might be very small.
Thus, it can be beneficial to add multiple routes with negative reduced cost to the {DRMP} in a single pricing step.
Moreover, the pricing problem is an NP-hard combinatorial optimization problem whereas the restricted master problem is usually a small linear program.
Thus, fast {heuristic pricing} can be highly beneficial for the overall runtime compared to exact pricing.
{Here, exact pricing refers to the task of solving~\eqref{eq:pricing} while pricing heuristics try to find a route $r$ with  $\bar{c}_{r}<0$
without any guarantees of success.}
Exact pricing is only required if the pricing heuristic does not return a route with negative reduced cost.
Hence, in the ideal case, we need to exactly solve the pricing problem {only once for solving the \eqref{eq:mp_cvrp}}.
The requirement of generating multiple solutions quickly renders QA a suitable heuristic for the pricing problem.

\paragraph{Rounded capacity cuts.}
Although the bound obtained from the LP-relaxation \eqref{eq:mp_cvrp} is relatively strong,
it is still insufficient for a practicable branch-and-bound approach on larger instances. 
Thus, in state-of-the-art algorithms, additional inequalities are added to the LP model in order to cut off fractional solutions and thereby to strengthen the bound.
Here, we focus on \emph{rounded capacity cuts}, which have proven particular useful in CVRP algorithms~\cite{Lysgaard_2004,Lysgaard_2003,Bixby1999}.
Given a set of customers $S\subseteq V$ with total demand $D(S)\coloneqq \sum_{i\in S} d_i$, the rounded capacity cut (RCC)
\begin{align}\label{eq:rcc}
	\sum_{r:r\cap S\neq \emptyset} y_r \geq \ceil*{\frac{D(S)}{K}}
\end{align}
is valid.
The RCC~\eqref{eq:rcc} states that at least $\ceil*{{D(S)}/{K}}$ vehicles are required to cover the demand in $S$.
The associated separation problems asks whether a \blue{customer subset $S\subseteq V$ exists whose corresponding} RCC~\eqref{eq:rcc} is violated by a given LP solution $y^*$ to the MP~\eqref{eq:mp_cvrp}.
The authors of~\cite{DIARRASSOUBA_2017_2} prove that the RCC separation problem is NP-complete.
Thus, polynomial-time algorithms exist for special cases only~\cite{DIARRASSOUBA_2017}.
Typically, practitioners employ heuristic separation procedures~\cite{Ralphs_2003,Lysgaard_2004,kim_2023_sep}.
{Separation heuristics try to find a violated RCC without any guarantees of success.}
Recently, the authors of~\cite{Pavlikov_2024} developed an exact separation routine based on integer programming.
{Exact separation routines are guaranteed to find a violated RCC if exists.}
A related class of cuts are the \emph{fractional capacity cuts},
\begin{align}\label{eq:fcc}
	\sum_{r:r\cap S\neq \emptyset} y_r \geq \frac{D(S)}{K}\ .
\end{align}
In contrast to RCCs, fractional capacity cuts can be separated in polynomial time~\cite{Toth2002}.

We note that adding RCCs to the master problem~\eqref{eq:mp_cvrp} alters the structure of the subproblem~\eqref{eq:pricing}.
In particular, each RCC~\eqref{eq:rcc}, defined by a customer subset $S\subseteq V$, adds a variable $\beta_S\geq0$ to the dual master problem~\eqref{eq:cvrp_dmp}.
The variable $\beta_S$ has an objective coefficient of $\ceil*{{D(S)}/{K}}$
and a coefficient of $1$ in each constraint whose corresponding route has a non-empty intersection with $S$.
Thus, the modified pricing problem reads
\begin{align}\label{eq:pricing_mod}
	\bar{c}_{\mathrm{min}} \coloneqq \min_{} \{c_r -\sum_{v\in r}\blue{\pi^*_v} -\sum_{S:S\cap r\neq \emptyset}\beta_S \mid r\in \mathcal{R}\}\ 
\end{align}
{where the second sum runs over all included RCCs which have a non-empty intersection with the route.}

\paragraph{Solving the integer program.}
In this work, we focus on solving the master problem~\eqref{eq:mp_cvrp} and the associated pricing and separation problems,
rather than on actual solutions to the CVRP integer program~\eqref{eq:mip}.
However, it is straightforward to extend our solution method for the linear relaxation~\eqref{eq:mp_cvrp} to an exact algorithm for the integer program~\eqref{eq:mip}.
{Although not part of this work, w}e briefly describe how such integer solutions can be obtained in principle.
Using column generation and RCC separation as an algorithm for the linear relaxation~\eqref{eq:mp_cvrp}, one can build a branch-and-bound framework
which solves the integer program to proven optimality.
This approach is called branch-price-and-cut and sketched in Figure~\ref{fig:concept_bpc}.
\blue{Branch-price-and-cut first solves the linear relaxation~\eqref{eq:mp_cvrp} by column generation
and tightens the bound by cutting plane separation.
If the relaxation solution is integral, the algorithm terminates.
Otherwise, two subproblems are created by restricting a variable which attains a fractional value to either 0 or 1.
Then, the linear relaxations of the subproblems are solved and tightened by cutting plane separation.
The process repeats, keeping track of a global upper and lower bound until they match.
Some care must be taken when adding branching rules and cutting planes to the master problem since they alter the structure of the pricing problem~\cite{Bixby1999}}.
We remark that another common heuristic approach is to apply column generation only
in the root node of the branch-and-bound tree.
Then, the restricted IP is tractable for a standard branch-and-cut solver.
This approach is known as price-and-branch and often yields close-to-optimal solutions~\cite{Desrosiers2011}.

\section{QUBO Models for Pricing and Separation}\label{sec:qvrp_qubos}
We require QUBO models for both the pricing problem~\eqref{eq:pricing} and the separation problem~\eqref{eq:rcc}
in order to employ QA as a heuristic for these problem.
We seek for QUBO models of both low dimension and low density.

\paragraph{Pricing model.}
In the following, we develop a QUBO model of the CPCTSP~\eqref{eq:pricing}.
Ref.~\cite{Lucas_2014} proposes a QUBO formulation for the related TSP.
We restate this QUBO model for the TSP and extend it to the CPCTSP.
To this end, we consider a complete graph $K_n=(V,E)$ together with distances $t_{uv}>0$ for $uv \in E$, which we assume to satisfy the triangle inequality.
The TSP asks for a \blue{route} of minimal length visiting all vertices of the graph.
Here, we define a \blue{route} as a sequence of vertices of length $n$ and call the indices in such a sequence \emph{time steps}.
We introduce a binary variable $x_{v,i}$ which
takes a value of~$1$ if and only if the vertex $v \in \{1,\dots,n\}$ is visited in time step $i\in \{1,\dots n\}$ in the \blue{route}.
Then, the TSP can be modeled via the QUBO cost function
\begin{align}
	C^{\mathrm{TSP}}(x)&=\sum_{(u,v)\in V\times V}\left(t_{uv}\sum_{i=1}^{n}x_{u,i}x_{v,i+1}\right)
	+P \sum_{v=1}^{n}\left( \sum_{i=1}^{n} 1-x_{vi}\right)^2
	+ P\sum_{i=1}^{n}\left( \sum_{v=1}^{n} 1-x_{vi}\right)^2 \,. \label{eq:TSP}
\end{align}
In the first sum, we identify $x_{v,n+1}$ with $x_{v,1}$.
The first term models the \blue{route} length.
The last two terms ensure feasibility by enforcing that each vertex is
visited exactly once (second term) and that exactly one vertex is
visited in each time step (third term).
A sufficient condition for $P$ ensuring that all feasible solutions have smaller costs than all infeasible ones is, for example, $P > T_n$ where $T_n$ is the sum of the $n$ largest distances.
Finally, we note that, without loss of generality, one can fix an arbitrary start vertex, reducing the number of variables from $n^2$ to $(n-1)^2$.

We extend Model~\eqref{eq:TSP} to a model for the CPCTSP~\eqref{eq:pricing}.
To this end, we first define some required quantities.
Given a CPCTSP instance $(V^0,t,D,K,\pi)$, defined by a CVRP instance $(V^0,t,D,K)$ and prices $\pi_{{v}}\geq 0$ for all ${{v}} \in V$,
we divide the capacity and all demands by their greatest common divisor $g = \gcd(K,D)$, defining
\begin{align}
	\tilde{K} \coloneqq K/g\ ,\qquad \tilde{D} \coloneqq \set{\blue{\underbrace{ d_{{v}}/g}_{\eqqcolon \tilde{d}_v}}\mid {{v}} \in V}\ .
\end{align}
Clearly, the resulting CPCTSP instance $(V^0,t,\tilde{D},\tilde{K},\pi)$ is equivalent to the original instance $(V^0,t,{D},{K},\pi)$.
{However, down-scaling the parameters reduces the number of required variables in the QUBO model.}
We set $\tilde{d}_{\mathrm{min}}\coloneqq \min (\tilde{D})$.
Additionally, we define $m\leq n$ as the largest number of customers a feasible route can contain.
The value of $m$ can be computed efficiently by simply sorting the demands in {ascending} order
{and cutting of cutting off when the collected demands exceed $K$}.
Finally, we combine prices and distances in the modified distances
\begin{align}
	\tilde{t}_{uv}  \coloneqq
	\begin{cases} 
		t_{uv}-(\blue{\pi^*_u}+\blue{\pi^*_v})/2 &u\neq v\\	
		0 &u = v\ .
	\end{cases}
\end{align}
{The modified distances allow us to model the reduced costs $\bar{c}_r$ of a route $r$ as a function of the route's edges since we have $\sum_{uv\in r}\tilde{t}_{uv}=c_r -\sum_{v\in r}\blue{\pi^*_v}=\bar{c}_r$.}
Note that $\tilde{t}_{uv}$ can attain negative values.
Having defined the preliminaries, we now introduce the following sets of binary variables.
\begin{itemize}
	\item Variable $x_{vj}$ indicates whether the \blue{route} visits vertex (customer or depot) $v\in V^0$ at time step $j \in \{1,\dots, m\}$.
	\item Variable $y_{v}$ indicates whether the \blue{route} visits customer $v\in V$.
	\item The variables $w_k$ for $k\in\{0, \dots, M-1\}$ form a binary encoding of the route's demand \blue{\cite{delagrandrive2019knapsackproblemvariantsqaoa,coffey2017adiabaticquantumcomputingsolution,quintero2021characterizing,Lucas_2014}}.
	Here, we set $M\coloneqq \ceil*{\log_2(\tilde{K}-\tilde{d}_{\mathrm{min}}+1)}$.
	Then, the route has total demand $\tilde{d}_{\mathrm{min}} + \sum_{k=0}^{M-1}2^kw_k$.
	We note that there is some degeneracy, i.e., multiple binary vectors $w$ encode the same demand, if $\tilde{K}-\tilde{d}_{\mathrm{min}}+1$ is not an integer power of 2.
\end{itemize}
With this, we define the QUBO model by
\begin{align}
	C^{\mathrm{CPCTSP}}(x,y,w) \coloneqq C^{\mathrm{CPCTSP}}_1(x,y,w)+P\cdot C^{\mathrm{CPCTSP}}_2(x,y,w)\label{eq:qubo_costfunction_cpctsp}
\end{align}
where
\begin{subequations}\label{eq:qubo_costfunction_cpctsp1}
	\begin{align}
		&C^{\mathrm{CPCTSP}}_1(x,y,w) \coloneqq\\
		&\sum_{(u,v)\in V^0\times V^0}\tilde{t}_{uv}\sum_{j=\blue{1}}^{\blue{m-1}}x_{u,j}x_{v,j+1}\label{eq:qubo_cpctsp_len}  \\
		&+\sum_{v\in V}\tilde{t}_{0v} x_{v,1}\label{eq:qubo_cpctsp_len_first}           \\
		&+\sum_{v\in V}\tilde{t}_{v0} x_{v,\blue{m}}\label{eq:qubo_cpctsp_len_last}     
	\end{align}
\end{subequations}
and
\begin{subequations}\label{eq:qubo_costfunction_cpctsp2}
	\begin{alignat}{2}
		 &C^{\mathrm{CPCTSP}}_2(x,y,w) \coloneqq\\
		 & \sum_{j = 1}^{m}\left(1-\sum_{v \in V^0} x_{vj} \right)^2\label{eq:qubo_cpctsp_time}                   \\
		 &+\sum_{v \in V}\left(y_{v} - \sum_{j = 1}^m x_{vj} \right)^2\label{eq:qubo_cpctsp_node}            \\
		 &+\Bigg(\tilde{d}_{\mathrm{min}} + \sum_{k=0}^{M-2}2^kw_k 
		 + \left(\tilde{K} - \tilde{d}_{\mathrm{min}} - 2^{M-1} + 1\right)w_{M-1} 
		 - \sum_{v\in V} \blue{\tilde{d}}_v y_v \Bigg)^2\label{eq:qubo_cpctsp_cap}\ .
	\end{alignat}
\end{subequations}
The term~\eqref{eq:qubo_cpctsp_len} models the travel costs between consecutive customers.
Terms~\eqref{eq:qubo_cpctsp_len_first} and~\eqref{eq:qubo_cpctsp_len_last}
count the travel costs to the first and from the last customer, respectively.
Penalty term~\eqref{eq:qubo_cpctsp_time} ensures that either exactly one customer or the depot is visited in each time step.
Via penalty term~\eqref{eq:qubo_cpctsp_node}, we enforce that each customer is visited at most once.
Penalty term~\eqref{eq:qubo_cpctsp_cap} is satisfied if the route's demand does not exceed the capacity.
We remark that incorporating the prices in the modified distances $\tilde{t}$ instead of counting them via $y$-variables
bears the advantage that some modified distances can be zero, resulting in fewer quadratic terms.

The total number of variables in Model~\eqref{eq:qubo_costfunction_cpctsp} amounts to \[\mathrm{\#Var}(C^{\mathrm{CPCTSP}})=(n+1)m + n + \ceil*{\log_2(\tilde{K}-\tilde{d}_{\mathrm{min}}+1)}.\]
This value is upper bounded by
\begin{align}\label{eq:nvar_improved}
	\mathrm{\#Var}(	C^{\mathrm{CPCTSP}}) \leq n^2 + 2n + \ceil*{\log_2({{K}})} .
\end{align}
Assuming $M\in O(n)$, the number of non-zero coefficients for the
quadratic monominals in Model~\eqref{eq:qubo_costfunction_cpctsp} is
upper bounded by
\begin{align}\label{eq:nquad_improved}
	\mathrm{\#Quad}(	C^{\mathrm{CPCTSP}})  \leq 3n^3 + O(n^2)\, .
\end{align}
This can be seen as follows. If $M\in O(n)$, the term~\eqref{eq:qubo_cpctsp_cap} contributes at most $O(n^2)$ quadratic monomials.
The dominating contribution then comes from terms~\eqref{eq:qubo_cpctsp_len}, \eqref{eq:qubo_cpctsp_time} and~\eqref{eq:qubo_cpctsp_node}, each of which contains at most $n^3$ quadratic monomials.
{If the number of reachable customers scales as $m\in \Omega(n)$}, we conclude that the density of Model~\eqref{eq:qubo_costfunction_cpctsp} satisfies
\begin{align}\label{eq:dens_improved}
	\rho(C^{\mathrm{CPCTSP}}) \in O(n^{-1})\,,
\end{align}
which goes to zero for large $n$.
It can easily be verified that for the TSP model~\eqref{eq:TSP}, it also holds $\rho(	C^{\mathrm{TSP}}) \in O(n^{-1})$.
Thus, {if $m\in \Omega(n)$} our model is asymptotically as sparse as the TSP model.
Indeed, any model that uses variables $x_{vj}$ to indicate whether vertex $v$ is visited at time step $j$ needs 
at least $O(n^2)$ variables and $O(n^3)$ quadratic monomials to enforce that exactly one vertex is visited at each time step via terms of the form of~\eqref{eq:qubo_cpctsp_time}.
Thus, in this case, a density $O(n^{-1})$ is asymptotically optimal.
{We note that in scenarios where number of reachable customers $m$ does not grow with $n$, the density does not tend to zero.}
\blue{However, in this case, the asymptotic behavior of the density can be still be bounded by
\begin{align}
    \rho(C^{\mathrm{CPCTSP}}) &\in O\left( \frac{n^2m+m^2n+n^2}{n^2m^2} \right) = O\left(m^{-1}+n^{-1}\right)\,.
\end{align}
}
{In Figs.~\ref{fig:v_vars_pricing} and \ref{fig:dens_pricing}, we calculate the actual numbers of variables and the densities
for all instances from the CVRPLIB~\cite{Uchoa2017} with up to 1,000 customers.
\blue{The CVRPLIB is a library of established CVRP benchmark instances.}
We observe that for classically hard instances with several hundreds of customers, the pricing QUBOs have in the order of $10^4$ variables and densities typically well below 0.2.
In contrast, the annealing device used in this work has roughly 5,000 physical qubits and a density of 0.003~\cite{Dwave_advantage}.}
\begin{figure}
    \centering
    \includegraphics[height=4.2cm]{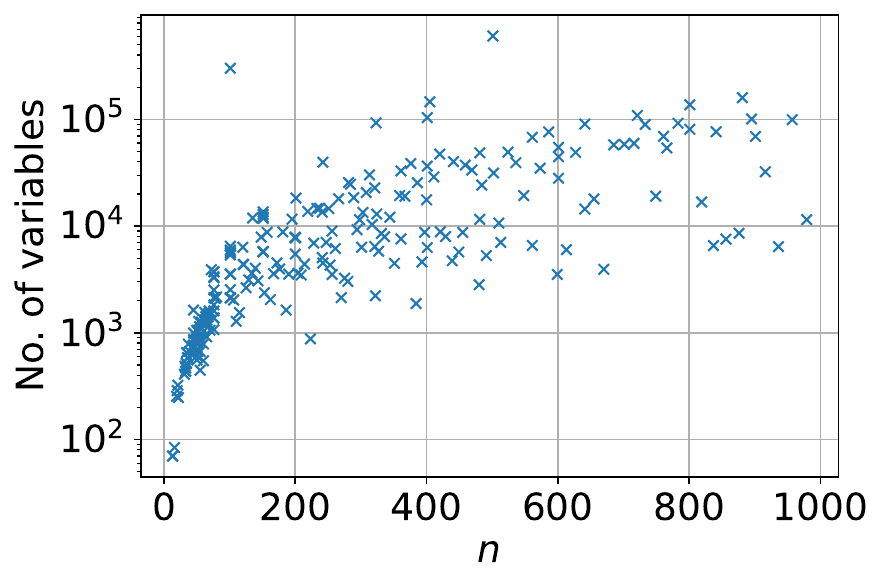}
\caption{Number of variables in the pricing QUBO for all instances from the CVRPLIB~\cite{Uchoa2017} with up to 1,000 customers.
\label{fig:v_vars_pricing}}
\end{figure}
\begin{figure}
    \centering
    \includegraphics[height=4.2cm]{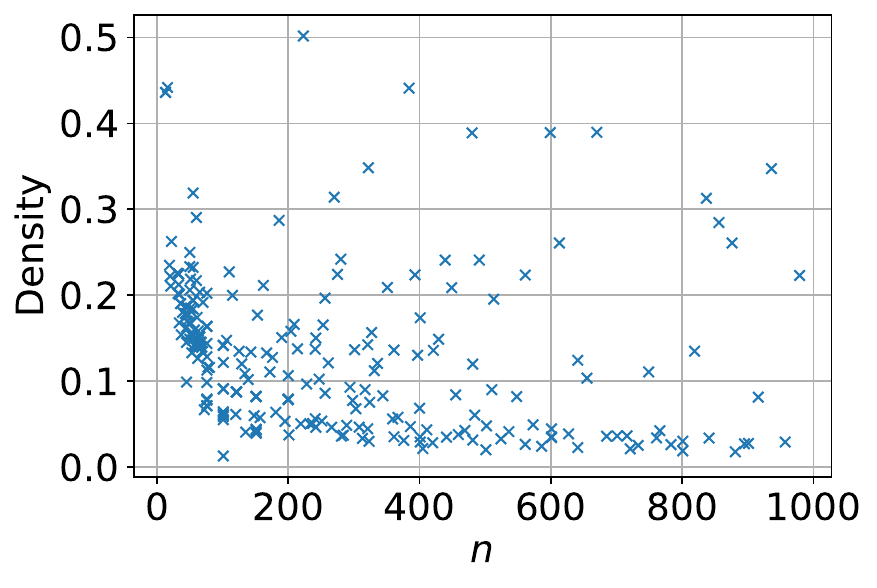}
\caption{Density of the pricing QUBO for all instances from the CVRPLIB~\cite{Uchoa2017} with up to 1,000 customers.
\label{fig:dens_pricing}}
\end{figure}

{We refer to a solution $(x,y,w)$ of Model~\eqref{eq:qubo_costfunction_cpctsp} with $C^{\mathrm{CPCTSP}}_2(x,y,w) > 0$ as \textit{infeasible}.}
In the following, we derive a sufficient
condition for a lower bound on the value of $P$ in Model~\eqref{eq:qubo_costfunction_cpctsp} such that an infeasible solution will always be more costly than a feasible one.
\begin{lemma}\label{lem:penalty_pricing}
Let $(x,y,w)$ be a binary vector such that $C^{\mathrm{CPCTSP}}_2(x,y,w)=0$
and let $(x',y',w')$ be a binary vector such that $C^{\mathrm{CPCTSP}}_2(x',y',w') > 0$.
If
\begin{align}\label{eq:PenaltyRatio}
	P > (n+1) \max_{}\set{\blue{\tilde{t}_{uv}} \mid \blue{uv \in V^0\times V^0}} + \sum_{v \in V}\blue{\pi^*_v},
\end{align}
then $C^{\mathrm{CPCTSP}}(x',y',w')>C^{\mathrm{CPCTSP}}(x,y,w)$.
\end{lemma}
\begin{proof}
For clarity, we omit the superscript $\mathrm{CPCTSP}$ in this proof.
We consider a feasible solution $(x,y,w)$ to \eqref{eq:qubo_costfunction_cpctsp}.
Thus, it holds $ C_2(x,w,y)=0$.
Further, we consider an infeasible solution $(x',y',w')$ to \eqref{eq:qubo_costfunction_cpctsp}.
That is, we have  $ C_2(x',y',w') > 0$.
Denote by $\Delta C \coloneqq  C(x',y',w')-C(x,y,w)$ the difference in costs.
We {require} that $(x',y',w')$ has a larger cost than $(x,y,w)$,
\begin{align}
	\Delta C = \Delta C_1 + P\Delta C_2 > 0\, .
\end{align}
Since  $ C_2(x',y',w') > 0$ and  $ C_2$ only has integer coefficients, we have $\Delta C_2 \geq 1$.
Hence,
\begin{align}
	 \Delta C_1 + P \Delta C_2 \geq \Delta C_1+P\ .
\end{align}
Now, we ask for the minimum value $\Delta C_1$ can attain.
A trivial bound on this is
\begin{align}
	\Delta C_1 \geq -\left((n+1) \max_{}\set{\blue{\tilde{t}_{uv}} \mid \blue{uv \in V^0\times V^0}} + \sum_{v \in V}\blue{\pi^*_v} \right).
\end{align}
The first term estimates the maximum loss in \blue{route} length,
whereas the second term estimates the maximum gain in price collection.
The statement of the lemma follows.
\end{proof}
{We note that, in general, the condition on $P$ in Lemma~\ref{lem:penalty_pricing} is sufficient but not necessary. Thus, it might be possible to improve the bound on $P$.}

\blue{
In our experiments, we compare heuristic pricing via QA and simulated annealing to exact pricing via integer programming.
To this end, we require an IP model for the pricing problem.
We could simply translate the QUBO~\eqref{eq:qubo_costfunction_cpctsp} into an
equivalent integer quadratic program.
Therein, the (quadratic) objective is given by~\eqref{eq:qubo_costfunction_cpctsp1}.
The (linear) constraints are constructed by
setting the square roots of each individual 
quadratic parenthesis in \eqref{eq:qubo_cpctsp_time}, \eqref{eq:qubo_cpctsp_node}, \eqref{eq:qubo_cpctsp_cap} to zero.
Then, the quadratic objective can further be linearized using
auxiliary variables, which most modern branch-and-cut solvers handle automatically~\cite{Gurobi}.}
\blue{However, solving a naively linearized quadratic model by branch-and-cut is usually
slower than solving an IP model tailored for branch-and-cut solvers.
Using a tailored IP model is thus a more
realistic approach than solving the (linearized) QUBO~\eqref{eq:qubo_costfunction_cpctsp} by a branch-and-cut solver.}
\blue{Hence, instead of linearizing QUBO~\eqref{eq:qubo_costfunction_cpctsp}}, we adopt the IP model for the CPCTSP from Ref.~\cite{Bixby1999}, which we restate here for completeness.
For a given CPCTSP instance defined by a CVRP instance $(V^0,t,D,C)$ together with prices $\set{\blue{\pi^*}_{{v}} \mid {{v}} \in V}$,
we denote by $K_{n+1}=(V^0,E)$ the complete graph with vertices $V^0$ and undirected edges $E$.
We introduce a binary variable $x_{{uv}}$ for each edge ${uv} \in E$, indicating whether edge ${uv}$ is in the \blue{route} ($x_{{uv}}=1$) or not ($x_{{uv}}=0$).
Analogously, binary variables $y_{{v}}$ for ${v} \in V^0$ indicate whether vertex ${v}$ is in the \blue{route} ($y_{{v}}=1$) or not ($y_{{v}}=0$).
With this notion, the IP model reads
\begin{subequations}\label{eq:bixby}
	\begin{alignat}{6}                                                              
		 & \min_{x,y}         &  & \sum_{{v} < {u}}t_{{uv}}x_{{uv}}-\sum_{{v} \in V}\pi_{{v}}y_{{v}}              &  & \label {eq:obj_bixby}                                          \\
		 & \mathrm{s.t.}\quad &  & \sum_{{v} \in V}x_{{uv}}                                 = 2y_{{v}}      &  & \quad\forall {u} \in V^0 \label{eq:cycle_bixby}                  \\
		 &                    &  & y_{0}                                                = 1         &  & \label{eq:depot_bixby}                                         \\
		 &                    &  & \sum_{{v} \in V}d_{{v}} y_{{v} }                               \leq K      &  & \label{eq:cap_bixby}                                           \\
		 &                    &  & \sum_{{u} \in S}\sum_{{v} \in V\setminus S}x_{{uv}}        \geq 2y_{{w}}   &  & \quad\forall S \subset V\ \forall {w} \in S\ \label{eq:sec_bixby} \\
		 &                    &  & x_{{uv}}                                               \in \{0,1\} &  & \quad\forall {uv} \in E \label{eq:binx_bixby}                    \\
		 &                    &  & y_{{v}}                                                \in \{0,1\} &  & \quad\forall {v} \in V\ .\label{eq:biny_bixby}
	\end{alignat}                                                                   
\end{subequations}
Constraints~\eqref{eq:cycle_bixby} require that every vertex in the \blue{route} has two adjacent edges while constraint \eqref{eq:depot_bixby} ensures that the depot is in the \blue{route}.\noeqref{eq:obj_bixby,eq:binx_bixby,eq:biny_bixby}
Constraint~\eqref{eq:cap_bixby} guarantees that the \blue{route} does not exceed the capacity.
Constraints~\eqref{eq:sec_bixby} are \emph{subtour elimination constraints} (SECs).
They enforce that the \blue{route} consists of a single cycle rather than a set of several disjoint cycles (subtours).
Clearly, there are exponentially many SECs.
To circumvent enumeration, SECs are typically added via a separation algorithm~\cite{Bixby1999}.
Checking if a violated SEC~\eqref{eq:sec_bixby} exists amounts to
solving at most $|V|$ many minimum-cut problems.
This can be done efficiently by a combinatorial algorithm~\cite{Corman2009}.

\blue{Finally, we remark that, if RCCs~\eqref{eq:rcc} are added to the \eqref{eq:mp_cvrp}, the modified pricing problem~\eqref{eq:pricing_mod} needs to be modeled as a QUBO.
Although not part of this work, we briefly sketch how this can be achieved.
For each RCC~\eqref{eq:rcc} defined by a customer subset $S\subseteq V$, we
add a binary variable $z_S$ to the pricing QUBO~\eqref{eq:qubo_costfunction_cpctsp}.
The variable indicates if the route intersects with $S$ or not.
This can be enforced via a penalty term of the form $\sum_{v\in S}z_S(1-y_v)$.
The coefficient of $y_S$ in $C_1^{\mathrm{CPCSTSP}}$ is $-\beta_S$.}

\paragraph{Separation model.}
Apart from generating columns via QA, we also aim at separating RCCs \blue{\eqref{eq:rcc}} via QA.
To this end, we model the separation problem for fractional capacity cuts~\eqref{eq:fcc}, i.e.,
\begin{align}\label{eq:fcc_sep}
	\min_{S\subseteq V} \sum_{r:r\cap S\neq \emptyset}y^*_r-\frac{D(S)}{{K}}\,,\tag{\blue{FCC}}
\end{align}
as a QUBO problem.
\blue{Recall that $y^*$ is an optimum solution to the~\eqref{eq:mp_cvrp}.}
Then, for all solutions $S\subseteq V$ returned by the QA algorithm, we check if the corresponding RCC~\eqref{eq:rcc} is violated.
{This procedure allows a more compact QUBO model than directly modeling the RCC separation problem~\eqref{eq:rcc} as a QUBO\blue{, which would require additional auxiliary variables to model the rounding function}.}
For the QUBO model, we introduce two sets of binary variables.
\begin{itemize}
	\item $x_{{v}}\in \{0,1\}$ encodes ${{v}}\in S$ for all ${{v}}\in V$.
	\item $w_r$ encodes $r\cap S \neq \emptyset$ for all $r \in r \in \blue{\mathcal{R}^*}$
\end{itemize}
\blue{where $\mathcal{R}^*\coloneqq \{r \in {\mathcal{R}} \mid y^*_r > 0\}$.}
With this, we define the QUBO model by
\begin{subequations}\label{eq:rcc_qubo}
	\begin{align}
		C^\mathrm{RCC}(x,w) & \coloneqq
		C^\mathrm{RCC}_1(x,w) + P \cdot C^\mathrm{RCC}_2(x,w)\\&=
		\sum_{r \in \blue{\mathcal{R}^*}}  y^*_r w_r - \frac{1}{{K}}\sum_{{{v}} \in V} d_{{v}} x_{{v}}\label{eq:qubo_rcc_cost}\\
		&+P\sum_{r \in \blue{\mathcal{R}^*}} \sum_{{{v}}\in r} x_{{v}} - w_r x_{{v}}\label{eq:qubo_rcc_feas}\ .
	\end{align}
\end{subequations}
The (linear) cost term~\eqref{eq:qubo_rcc_cost} minimizes~\eqref{eq:fcc_sep}.
The (quadratic) penalty term~\eqref{eq:qubo_rcc_feas} ensures that if
$x_i=1$, then $w_r = 1$ for all $r\in R$ and $i\in r$.
The number of variables in Model~\eqref{eq:rcc_qubo} amounts to
\begin{align}\label{eq:nvar_rcc}
	\mathrm{\#Var}(	C^{\mathrm{RCC}})=n+|\blue{\mathcal{R}^*}| \leq 2n\ .
\end{align}
{In the last inequality we have used the fact that in any basis solution of the simplex algorithm, the number of non-zeros is bounded by the number of constraints.}
The number of non-zero quadratic coefficients
in Model~\eqref{eq:rcc_qubo} is bounded by
\begin{align}\label{eq:nquad_rcc}
	\mathrm{\#Quad}(	C^{\mathrm{RCC}}) \leq nm\, .
\end{align}
{In general we have $m\in O(n)$ and thus} the density of Model~\eqref{eq:rcc_qubo} satisfies
\(
	\rho(	C^{\mathrm{RCC}}) \in O(1)\,,
\)
which
{does not allow a conclusion about the sparsity in practice}.
{However, in scenarios where $m\in O(1)$, we have $\rho(	C^{\mathrm{RCC}}) \in O(n^{-1})$}.
{Moreover}, the number of variables is considerably smaller than in the pricing model.
{We note that, in contrast to the pricing problem QUBO~\eqref{eq:qubo_costfunction_cpctsp}, the number of variables and the non-zero quadratic coefficients in the RCC QUBO model~\eqref{eq:rcc_qubo} are not fully determined by the CVRP instance but also depend on the current fractional solution to the~\eqref{eq:mp_cvrp}.}

We now derive a sufficient condition for the value of the penalty factor $P$ which ensures
that an optimum solution to Model~\eqref{eq:rcc_qubo} is indeed an optimum solution to the separation problem~\eqref{eq:fcc_sep}.
\begin{lemma}
	If $P>1$,
	then
    {for each binary vector $(x,w)$ with $C^\mathrm{RCC}_2(x,w)>0$}
    there exists {another} binary vector $(x',w')$ such that $C^\mathrm{RCC}(x',w')<C^\mathrm{RCC}(x,w)$.
\end{lemma}
\begin{proof}
	Since $C^\mathrm{RCC}_2(x,w)>0$, there exists  $r \in R$ and $i \in r$
	such that $x_i = 1$ and $w_r = 0$.
	We derive $(x',w')$ from $(x,w)$ by changing the value of $w_r$ from $0$ to $1$.
	Using $y^*_r \leq 1$, the statement follows.
\end{proof}
Above lemma implies that if $P> 1$, an optimum solution to~\eqref{eq:rcc_qubo} encodes an optimum solution to \eqref{eq:fcc_sep}.
{Based on practical experience, we use $P=2$ in our experiments.}

\section{Experimental Results}\label{sec:qvrp_exp}
In this section, we  evaluate the developed QUBO models for pricing and separation on actual quantum hardware.
\blue{Code and data is available at~\cite{qvrp}.}
Our goal is to identify acceleration potential for branch-price-and-cut when quantum technology advances.
In literature, the simulated annealing (SA) meta-heuristic is often employed as a benchmark for quantum annealing.
The reason is that QA can be viewed as a quantum version of SA.
Also in this work, we compare QA to SA.
However, we additionally evaluate classical algorithms, specifically tailored to the pricing and separation problem.
This resembles a more realistic classical scenario, where problem-specific algorithms often outperform meta-heuristics.

The quantum annealing hardware in our experiments is a Dwave advantage {system 4.1} processor~\cite{Dwave_advantage}
accessed via the {Python interface \textit{dwave-system}, version 1.28~\cite{DwaveOcean}}.
{The working graph of the processor has a Pegasus topology and its density amounts to roughly 0.003~\cite{Dwave_advantage}.}
{For QUBO modeling, we use the Python package \textit{dimod}, version 0.12.}
Based on practical experience, we set the annealing time to {5}~µs and the sample size to $5,000$.
{Compared to the default annealing time of $20$~µs, we observed a higher solution quality on average.}
\blue{We evaluate each of the $5,000$ samples individually.}
{All remaining settings were left to default.}
{In particular, we use the default embedding algorithm \texttt{minorminer}~\cite{cai2014practical}.}
{However, we re-use the computed embeddings for pricing QUBOs in successive iterations on the same CVRP instance.}
{This is possible since the connectivity of the pricing QUBOs only depends on the given \blue{original} CVRP instance.}
\blue{The default settings calculate the QA chain strength by the method \texttt{uniform\_torque\_compensation} and the unembeddding is done by majority vote~\cite{DwaveOcean}.}
The classical hardware is a standard laptop.
We employ an open-source SA algorithm~\cite{DwaveOcean},
which we run with default settings.
To ensure a fair comparison, we use the same QUBO model and sample size as in QA.
For solving LPs and IPs, we use the branch-and-cut solver Gurobi~\cite{Gurobi}
{ version 12 with default settings, accessed via its Python interface.}

We consider several CVRP instances from literature, available in the \emph{CVRPLIB}~\cite{Uchoa2017} \blue{and the TSPLIB95~\cite{reinelt1995tsplib95} (instance \qm{eil7})}.
{In particular, we select the 11 smallest instances due to hardware limitations.}
Additionally, we study instances arising from a real-world application.
In this application, vehicles correspond to waste-collecting trucks, while customers and demands correspond to waste containers of different fill levels.
{The instances were obtained from a distance matrix containing real distances between 29 containers. To generate a CVRP instance, we choose a subset of the containers, apply unit demands and truck capacities between 5 and 7.
We consider these additional instances to have more instances which are small enough for execution on quantum hardware.}
The names of the waste-collection instances start with a \qm{W}.
Table~\ref{tab:instances} summarizes the instance data.
\begin{table}[t]
	 \centerline{
		\begin{tabular}{lrrrrrrr}
			\toprule
			Instance  & $n$ & $c^*$ & $c^*_{\mathrm{MP}}$ & \multicolumn{2}{c}{\blue{$\mathrm{\#Vars}$}} & \multicolumn{2}{c}{$\mathrm{\rho}$}\\
            \cmidrule(lr){5-6}\cmidrule(lr){7-8}
            &&&&\eqref{eq:qubo_costfunction_cpctsp}&\eqref{eq:rcc_qubo}&\eqref{eq:qubo_costfunction_cpctsp}&\eqref{eq:rcc_qubo}\\
			\midrule
            eil7 & 7 & 114 & 114.0 &      28&-&0.53 & - \\  
			E-n13-k4 & 13 & 290 & 264.0   &  70&21.7& 0.44 &  0.15\\ 
			W-n15-k2 & 15 & 76 & 76.0  &     122&-&0.34 &  -\\ 
			P-n16-k8 & 16 & 450 & 441.0   &  84&26.3& 0.087&  0.066\\ 
			W-n18-k3 & 18 & 108 & 99.27   &  128 &24.3&0.37 & 0.13\\ 
			P-n19-k2 & 19 & 212 & 204.29  &  -& 23.5& -&  0.14\\  
			W-n20-k4 & 20 & 110 & 104.0   &  142 &26.7&0.37 & 0.12\\  
			P-n20-k2 & 20 & 216 & 212.0   &  -&23.0& -&  0.15\\  
			P-n22-k8 & 22 & 603 & 589.67  &  158&31.0& 0.23 &  0.051\\ 
			E-n22-k4 & 22 & 375 & 373.71  &  -&29.0& - &  0.092\\ 
			P-n23-k8 & 23 & 529 & 521.54  &  165&35.3& 0.22 &  0.056\\  
			W-n28-k6 & 28 & 156 & 154.3   &  169&36.5& 0.40 &   0.061\\  
			E-n30-k3 & 30 & 534 & 484.1   &  -&40.0& - & 0.11 \\ 
			E-n31-k7 & 31 & 1815 & 1188.13&  -&58.0& - &  0.089 \\ 
			A-n32-k5 & 32 & 784 & 758.43  &  -&49.7&-  & 0.11 \\ 
			\bottomrule
		\end{tabular}
				\caption{Instance data.
			In column \qm{Instance}, we give the instance name.
			The column \qm{$n$} reports the number of customers.
			In columns $c^*$ and $c^*_{\mathrm{MP}}$, we give the value of an optimum CVRP solution and
			the optimum value of the linear relaxation~\eqref{eq:mp_cvrp}, respectively.
            Columns \blue{$\mathrm{\#Vars}$ and} ${\rho}$ report the (average) \blue{dimension} and density of the pricing \blue{\eqref{eq:qubo_costfunction_cpctsp}} and separation \blue{\eqref{eq:rcc_qubo}} QUBOs, respectively.
                        \blue{Missing variable numbers and densities are due to failed embeddings (pricing) or a vanishing integrality gap (separation).}
}
		
		\label{tab:instances}	
        }
\end{table}

\paragraph{Pricing.}
First, we compare {heuristic} pricing by QA to {heuristic pricing by} SA and to a classical exact pricing algorithm.
Fig.~\ref{fig:concept_pricing} visualizes our evaluation workflow.
We first initialize the~\eqref{eq:cvrp_dmp} by all $n$ routes visiting exactly one customer.
Then, we iteratively solve the~\eqref{eq:cvrp_dmp} until the considered pricing heuristic, either SA or QA, does not return any routes with negative reduced costs.
In this case, the pricing heuristic is \blue{not invoked again but} replaced by {the} exact pricing algorithm based on integer programming (IP).
{For each pricing method, we add all routes with negative reduced cost to the~\eqref{eq:cvrp_dmp}.}
The iteration between master and subproblem continues until no further columns with negative reduced costs exist.
At this point, the master problem~\eqref{eq:mp_cvrp} is solved to optimality.
We remark that using heuristics for pricing and adding multiple variables with negative reduced costs per pricing step are common techniques in modern branch-price-and-cut implementations~\cite{Pecin2017,Pessoa2019}. Nonetheless, exact pricing is required at least once in order to verify that no further variables with negative reduced costs exist.
For comparison, we additionally solve each instance solely by {the} exact IP pricing {algorithm employed when the heuristic failed}.
This allows us to quantify the effectiveness of the pricing heuristic by measuring the reduction in exact pricing calls.
The limitations imposed by the quantum hardware restricts the instance size in our pricing experiments to a maximum of $28$ customers.
{For larger instances, the default embedding strategy failed.}
{Moreover, embedding failed for the instances P-n19-k2, P-n20-k2 and E-n22-k4.}

We note that the considered instances are small compared to the capabilities of classical exact solvers and can be solved to optimality in the order of seconds.

\begin{figure}
    \centering
    \includegraphics[height=6cm]{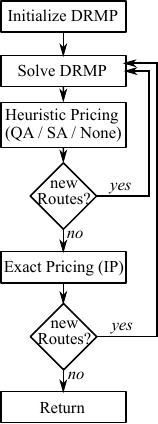}
    \caption{Schematic algorithm workflow for benchmarking pricing heuristics.
We compare exclusively exact pricing via integer programming (IP) to simulated annealing (SA) and quantum annealing (QA).
\label{fig:concept_pricing}}
\end{figure}
\begin{figure}
    \centering
    \includegraphics[height=4.2cm]{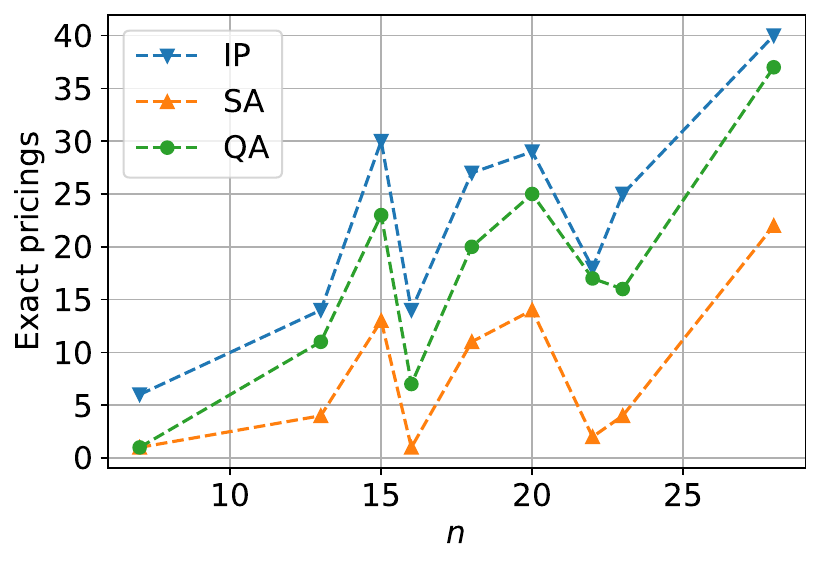}
    \caption{Number of expensive exact pricings in dependence of the instance size $n$ (lower is better).
\label{fig:res_pricing_ip}}
\end{figure}
\begin{figure}
    \centering
    \includegraphics[height=4.2cm]{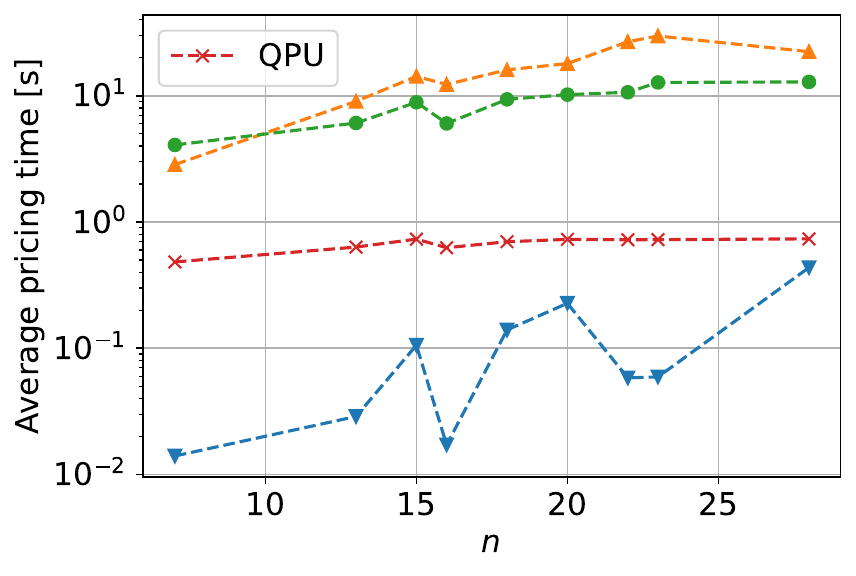}
    \caption{Average runtime per pricing problem in dependence of the instance size $n$.
The legend is depicted in Fig.~\ref{fig:res_pricing_ip}.
For QA, we also report the average time taken by the quantum processing unit (QPU).
\label{fig:res_pricing_time}}
\end{figure}
\begin{figure}
    \centering
    \includegraphics[height=4.2cm]{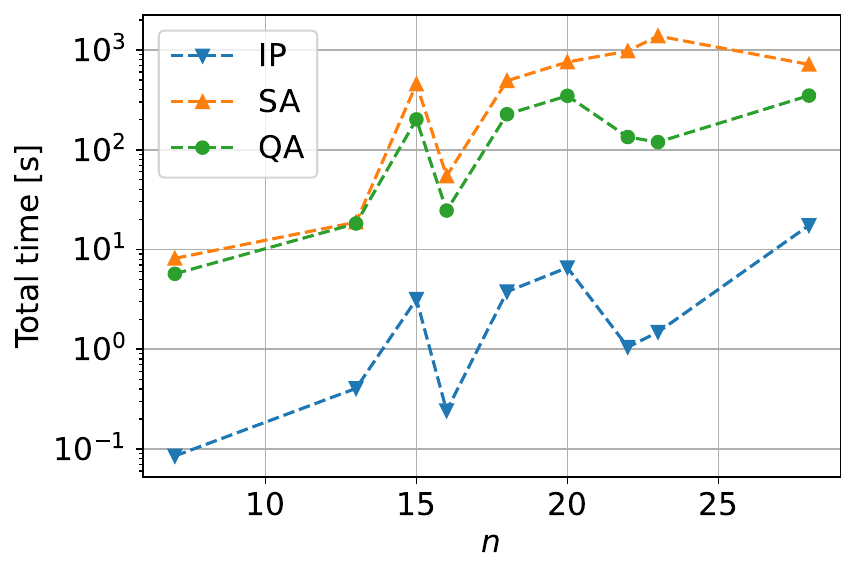}
    \caption{{Total runtime for solving the \eqref{eq:mp_cvrp} in dependence of the instance size $n$ for different pricers.
\label{fig:res_pricing_time_total}}}
\end{figure}

For QA and SA pricing, we use the QUBO model~\eqref{eq:qubo_costfunction_cpctsp} with penalty factors according to~\eqref{eq:PenaltyRatio}.
{In Table~\ref{tab:instances}, we report the actual densities of the pricing QUBOs.
We observe that the densities vary broadly between 0.09 and 0.44.
All densities are significantly larger than the device density of roughly 0.003.
Thus, embedding becomes prohibitive for larger instances.}

For exact pricing, we \blue{solve the} integer program~\blue{\eqref{eq:bixby}} to optimality via a branch-and-cut solver.
In our implementation of Model~\eqref{eq:bixby}, we add SECs via a callback routine whenever they are violated by an integer solution.
{We model and solve the IP~\eqref{eq:bixby} via the Python interface of the solver Gurobi~\cite{Gurobi}.
We implement the callback in Python using
the minimum-cut algorithm included in the Package NetworkX \blue{v3.0}~\cite{hagbger2008networkx}.}
We note that the IP model~\eqref{eq:bixby} for the CPCTSP differs significantly from the QUBO model~\eqref{eq:qubo_costfunction_cpctsp} since it is tailored for branch-and-cut solvers which require linear models.
Importantly, we return all routes with negative reduced cost found during the branch-and-cut process, not only the optimal solution.

In {Fig.}~\ref{fig:res_pricing_ip}, we compare the number of pricing subproblems which were solved exactly via integer programming.
As illustrated in Fig.~\ref{fig:concept_pricing}, exact pricing is only performed if the heuristic fails.
Avoiding exact pricing is crucial since the pricing problem is NP-hard and solving it exactly contributes significantly to the total runtime.
Indeed, exact pricing consumed over 99 \% of the total runtime for the experiments without a pricing heuristics.
{We thus consider the reduction of exact pricings as a measure of effectiveness of the heuristics.}
We observe that SA reduces
the number of exact pricing calls compared pricing exclusively via IP, often significantly, with an average reduction of 72~\%.
Also QA pricing reduces the number of exact pricings on all instances,
yet less significant than SA with an average reduction of 31~\%.
{We conclude that the solution quality of QA needs to improve in order
  to outperform classical pricing methods.} This is even more true
when considering modern implementations where IP pricing is typically
replaced as far as possible by fast heuristics. 

Additionally, we compare the average runtime per pricing problem in Fig.~\ref{fig:res_pricing_time}.
{The reported times are wall-clock times, including model building, embedding, sampling, un-embedding and solution construction.}
Here, we observe that IP outperforms both SA and QA by at least an order of magnitude across all instances.
We remark that all considered instances are small such that IP pricing is sufficiently fast.
However, the worst-case runtime of IP scales exponentially with problem size.
Thus, for larger instances where IP pricing is more expensive, both SA and QA pricing may lead to a runtime improvement.
Comparing QA and SA, we observe that QA is faster on all but the smallest instance.
Here, we remark that the major part of the {total runtime in} QA runtime is classical pre- and post-processing rather than time taken by the quantum processing unit (QPU).
In our experiments, the QPU time is roughly 10~\% of the total QA runtime on all instances.
From this, we conclude that QA runtime can be improved significantly by reducing classical overhead.

{
Finally, we compare the total runtime in Fig.~\ref{fig:res_pricing_time_total}.
The conclusions are similar to the average pricing runtime in Fig.~\ref{fig:res_pricing_time}.
QA is faster than SA on all instances by up to a factor of 10 ($n=23$).
Still, IP is at least an order of magnitude faster than QA.
IP, however, has an exponential worst-case runtime.
Thus, QA may outperform IP pricing when quantum hardware improves such that larger instances can be handeled.
}

To conclude, QA may be a valuable pricing heuristic but quantum
technology needs to improve considerably in order to outperform classical heuristics, even when comparing to meta-heuristics like SA,
which are typically less performant than problem-specific heuristics.
Indeed, in state-of-the art branch-price-and-cut implementations,
pricing is performed by specifically designed heuristics which are often very effective~\cite{Pecin2017,Pessoa2019}.
While SA outperforms QA in reducing the number of expensive exact pricing calls, QA often leads to shorter runtimes than SA.
Moreover, the major fraction of the QA runtime is consumed by classical computation.
Thus, QA bears the potential to significantly outperform SA in terms of runtime.
However, in order to solve instances where IP pricing is prohibitively
expensive, much larger quantum processors with a considerably increased connectivity are required.

\paragraph{Separation.}
Next, we benchmark RCC separation via QA against SA and an established RCC separation heuristic included in the \emph{CVRPSEP} library (CS)\blue{~\cite{Lysgaard_2003,Lysgaard_2004}}.
Fig.~\ref{fig:concept_sep} visualizes our evaluation procedure.
{We solve} the MP to optimality via column generation {and exact IP pricing with the same implementation as in the previous section.}
{Then}, we invoke a separation heuristic.
{While SA and QA are implemented with the same packages as in the pricing experiment, the RCC part of the CVRPSEP library is included in Python as a compiled C-library.}
As long as violated RCCs are found, we add them to the MP and re-optimize.
{Similar to the pricing experiments, we add all found violated RCCs.}
We consider only those instances from Table~\ref{tab:instances} which have a non-zero integrality gap.
The integrality gap is defined as \[g_\mathrm{LP}\coloneqq \frac{c^*-c^*_{LP}}{c^*}\in [0,1]\] with $c^*$ and $c^*_{LP}$ being the value of an optimum CVRP solution and an optimum solution to the relaxation \eqref{eq:mp_cvrp}, respectively.
The separation QUBO model~\eqref{eq:rcc_qubo} is smaller than the pricing QUBO model~\eqref{eq:qubo_costfunction_cpctsp}.
{Moreover, from Table~\ref{tab:instances}, we conclude that the separation QUBOs are sparser than the pricing QUBOs on average.}
This allows us to consider larger instances with up to $32$ customers.

\begin{figure}
    \centering
    \includegraphics[height=6cm]{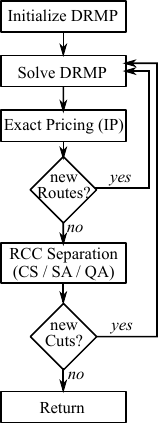}
    \caption{Schematic algorithm workflow for benchmarking separation heuristics.
	\label{fig:concept_sep}}
\end{figure}
\begin{figure}
    \centering
    \includegraphics[height=4.2cm]{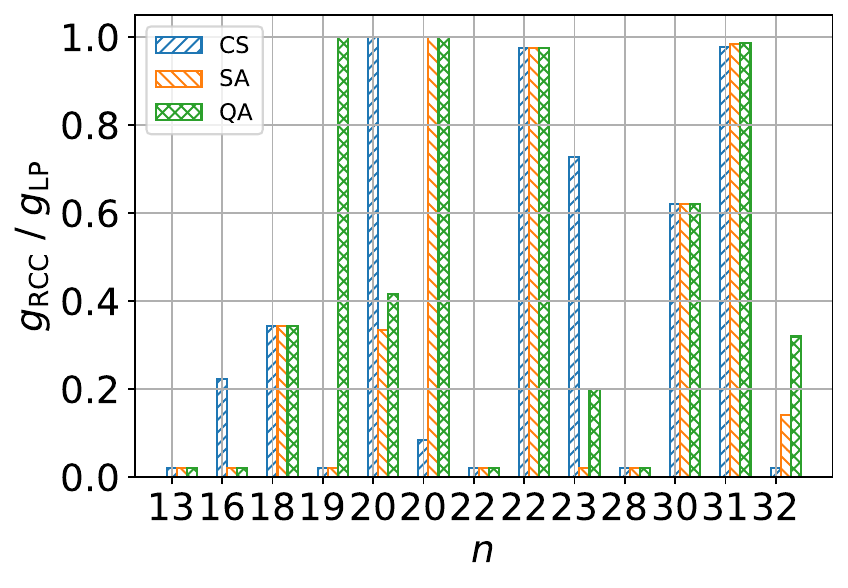}
    \caption{Decrease in integrality gap for different instance sizes. Lower is better.
	\label{fig:res_sep_gap}}
\end{figure}
\begin{figure}
    \centering
    \includegraphics[height=4.2cm]{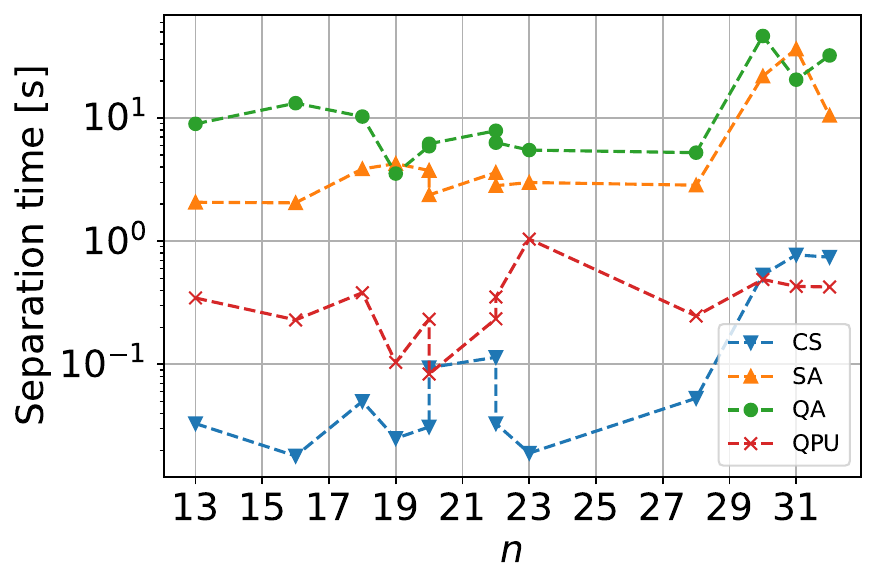}
    \caption{Total time spent on separation. The legend is depicted in Fig.~\ref{fig:res_sep_gap}. For QA, we also report the average time taken by the quantum processing unit (QPU).
	\label{fig:res_sep_time}}
\end{figure}

In Fig.~\ref{fig:res_sep_gap}, we visualize the decrease of the integrality gap after cutting plane separation.
To this end, we measure the ratio $g_\mathrm{RCC} / g_\mathrm{LP}$ where $g_\mathrm{RCC}$ is the integrality gap after separation.
Thus, a value of $0$ means that separation closed the integrality gap entirely, whereas a value of $1$ means that separation did not decrease the gap at all.
On {three} instances, QA separation yields smaller gaps than CVRPSEP.
However, on all instances SA delivers gaps at least as small as QA.
Similar to the pricing experiments, we conclude that QA separation bears potential but quantum
hardware needs to improve in order to outperform classical separation heuristics.

Finally, we compare the total time spent on RCC separation in Fig.~\ref{fig:res_sep_time}.
CVRPSEP is at least one order of magnitude faster than SA and QA on all instances.
Comparing SA and QA, we note that SA is faster on all but two instances
However, we remark that the pure QPU time consumed by QA is in the same range as the runtime of CVRPSEP.
As in our pricing experiments, we conclude that classical overhead in QA needs to decrease for a net runtime advantage.

\section{Conclusion}\label{sec:qvrp_concl}
In this work, we proposed using QA as a pricing and separation heuristic in a classical branch-price-and-cut algorithm for vehicle routing.
Although QA being a heuristic, the integration as a subroutine in a classical decomposition algorithm allows an exact solution of the input problem.
Moreover, the pricing and separation subproblems are particularly well suited for randomized heuristics like QA
since both require multiple, high-quality solutions quickly.
{Pricing and separation are standard routines in mixed-integer programming algorithms.
Thus, the concepts proposed in this work have many applications beyond vehicle routing.}

In order to apply QA, we developed QUBO models for the pricing and separation problem.
{Therein, we aimed at k}eeping both model dimension and density as low as possible {which }is crucial
for a successful implementation on existing quantum computers with restricted connectivity.
Here, developing even smaller and sparser QUBO models can further improve the performance.
For example, a technique known as domain-wall encoding could possibly further reduce density and dimension~\cite{Chancellor2019,Chen2021,Berwald2022}.
Although both dimension and density would not change asymptotically when using domain-wall encoding, a constant-factor reduction can be achieved which might result in practical benefits.
\blue{Furthermore, analyzing and adjusting the coefficient ranges of the QUBO models can improve the results from QA machines with limited precision.}

Having developed QUBO models for pricing and separation,
we integrated quantum annealing in {the pricing and separation part of }a branch-price-and-cut algorithm for the CVRP.
We then evaluated the algorithm on \blue{a small subset of} instances from literature and from a real-world application.
Our experiments revealed that quantum annealing is capable of reducing the number of exact pricing problem solutions.
As exact pricing problem solutions are computationally expensive, reducing their number is crucial for the overall runtime.
In cutting plane separation, we observed that quantum annealing yields smaller integrality gaps than a classical separation heuristic on some instances.
Small integrality gaps reduce the branch-price-and-cut runtime by limiting the number {of} branches.
{However, classical pricing and separation algorithms clearly outperform quantum annealing in our experiments.
Thus, our main conclusion is that quantum hardware needs to advance in order to achieve a net runtime advantage over classical methods.}
\blue{Fine-tuning QA parameters such as embedding strategy, chain strength and de-embedding may improve the performance of QA, but the
underlying hardware limitations clearly constitute the main bottleneck.}
In particular, larger quantum processors are required to reach instance sizes where classical methods become impractical.
Moreover, a large fraction of the quantum annealing runtime is consumed by classical computation.
Thus, accelerating classical pre- and post-processing can reduce the quantum annealing runtime significantly.
To summarize, our results indicate that QA can be a valuable pricing and separation heuristic if quantum hardware improves.
{However, a more exhaustive evaluation is required, including larger instances and more repetitions, to draw reliable conclusions for practical scenarios.}

A natural direction of further research is the implementation and evaluation of other quantum algorithms, for example QAOA.
{At the same time, developing and integrating further classical heuristics for the considered subproblems would enrich the comparison.}
{Furthermore, an integration and evaluation of our concepts
in a full branch-price-and-cut algorithm is required to asses their usefulness more holistically.
Therein, it can might even be beneficial to apply quantum algorithms to other sub-tasks like branching and node selection.}
In general, our work encourages the integration of fast quantum heuristics in established integer optimization algorithms if quantum technology advances.

\section*{Acknowledgment}
The authors thank Lilly Palackal, \blue{Infineon AG}, for providing instance data.
\blue{Furthermore, the authors thank the anonymous reviewers for their helpful comments.}
This research has been supported
by the Bavarian Ministry of Economic Affairs, Regional Development
and Energy with funds from the Hightech Agenda Bayern,
and
by the Federal Ministry for	Economic Affairs and Climate Action on the basis
of a decision by the German Bundestag through project QuaST.
	
	\bibliographystyle{unsrturl}  
	\bibliography{./bib_qvrp}
	
\end{document}